\documentclass[12pt]{amsart}
\usepackage[margin=1in]{geometry}
\usepackage{amsfonts}
\usepackage{amssymb}
\usepackage[dvips]{graphics}
\usepackage{epsfig}
\usepackage{euscript}
\usepackage{color}
\usepackage{bbm}
\usepackage[colorlinks]{hyperref}
\hypersetup{allcolors=blue}

\newtheorem{thm}{Theorem}[section]
\newtheorem{cor}[thm]{Corollary}
\newtheorem{lem}[thm]{Lemma}

\theoremstyle{definition}

\theoremstyle{remark}

\numberwithin{thm}{section}

\newcommand{\R}{{\mathord{\mathbb R}}}
\newcommand{\N}{{\mathord{\mathbb N}}}

\newcommand{\C}{{\mathord{\mathbb C}}}
\newcommand{\Z}{{\mathord{\mathbb Z}}}
\newcommand{\E}{{\mathord{\mathbb E}}}

\def\idty{{\mathchoice {\mathrm{1\mskip-4mu l}} {\mathrm{1\mskip-4mu l}} %
{\mathrm{1\mskip-4.5mu l}} {\mathrm{1\mskip-5mu l}}}}

\DeclareMathOperator{\I}{\mathbbm{1}}

\DeclareMathOperator{\tr}{tr}
\DeclareMathOperator{\Tr}{Tr}
\DeclareMathOperator{\Span}{span}

\DeclareMathOperator{\diag}{diag}
\begin{document}

\title[Entanglement Dynamics of Disordered Quantum XY Chains]{Entanglement Dynamics\\ of Disordered Quantum XY Chains}

\author[H. Abdul-Rahman]{Houssam Abdul-Rahman}
\address{Department of Mathematics\\
University of Alabama at Birmingham\\
Birmingham, AL 35294 USA}
\email{houssam@uab.edu}
\author[B. Nachtergaele]{Bruno Nachtergaele}
\address{Department of Mathematics, University of California, Davis,
Davis, CA 95616}
\email{bxn@math.ucdavis.edu}
\author[R. Sims]{Robert Sims}
\address{Department of Mathematics\\
University of Arizona\\
Tucson, AZ 85721, USA}
\email{rsims@math.arizona.edu}
\author[G. Stolz]{G\"unter Stolz}
\address{Department of Mathematics\\
University of Alabama at Birmingham\\
Birmingham, AL 35294 USA}
\email{stolz@math.uab.edu}

\date{}


\begin{abstract}
We consider the dynamics of the quantum XY chain with disorder under the general assumption that the expectation
of the eigenfunction correlator of the associated one-particle Hamiltonian satisfies a decay estimate typical of Anderson localization. 
We show that, starting from a broad class of product initial states, entanglement remains bounded for all times. For the XX chain,
we also derive bounds on the particle transport which, in particular, show that the density profile of initial states that consist of fully
occupied and empty intervals, only have significant dynamics near the edges of those intervals, uniformly for all times.
\end{abstract}

\maketitle

Keywords: XY Spin Chain, Disordered Systems, Quantum Entanglement, Many-Body Localization

\vskip .3cm

MSC: 82B44

%
%

\allowdisplaybreaks
\section{Introduction}\label{sec:intro}
\subsection{Motivation and Context}
In this work, we study the entanglement dynamics in a class of disordered quantum XY chains in the dynamical localization regime. We prove that for a large class of product initial conditions the bipartite entanglement satisfies a constant bound, independent of time and system size, i.e., the same
type of `Area Law' bound Hastings proved for gapped ground states in one dimension \cite{Hastings}. 
This is in agreement with the numerical predictions of \cite{Bardarsonetal} and it confirms the doubts expressed in \cite{chiara:2006} that the observed logarithmic growth of entanglement for short times holds for all times.
Disordered XY chains show many of the features usually associated with Many-Body-Localization (MBL) \cite{Baskoetal,PalHuse}. By exhibiting complete localization, the entanglement dynamics of the XY chain, however, appears to deviate from the logarithmic growth in time generically expected and observed numerically in other model systems \cite{chiara:2006,Znidaricetal,Bardarsonetal} and supported by renormalization group arguments \cite{vosk:2014}.

There are three major categories of dynamical behavior in extended quantum systems: ballistic, diffusive, and localized. It is tempting to associate these
three classes of transport behavior with three types of systems in which we expect them to occur: integrable homogeneous systems, generic non-integrable
but homogeneous systems, and strongly disordered systems. The integrable systems are clearly non-ergodic. The non-integrable systems are expected to be
generically ergodic. The ergodic properties of many-body systems with sufficiently random interactions are expect to display the syndrome called
Many-Body Localization (MBL) and the ergodic properties associated with MBL are currently a subject of debate
\cite{Baskoetal,PalHuse,OganesyanHuse,deroeck:2014,deroeck:2015,abanin:2015}.

To gain further insight in the nature of MBL, it is useful to study simple model systems. A number of interesting discoveries have been made in recent years
through numerical studies of the quantum Ising chain, the XY chain, and the XXZ chain \cite{Znidaricetal,Bardarsonetal}. E.g., it has become clear that the transport behavior of particles, energy, and entanglement do not always line up. E.g., it was found that in a particular quantum Ising chain with an external field that has non-vanishing $x$ and $z$ components, which is non-integrable, energy transport is diffusive but entanglement between two halves of the chain in a product initial state spreads ballistically \cite{kim:2013}. In the XXZ chain in a random field, a system that exhibits features of MBL, it was found numerically that entanglement grows to a saturation value which itself diverges with the length of the chain, except at the XY point, where it reaches a constant value independent of system size \cite{Bardarsonetal}. Arguments for a logarithmic growth in time have been given in \cite{Serbynetal2013a}. These arguments are based on the conjectured existence of an extensive number of localized conserved quantities many-body localized systems \cite{HuseOganesyan}. For the quantum Ising model in a random field, and under  some additional assumptions, an approach to a mathematical proof for the existence of such a set of local conserved quantities, called pseudo-spins, is described in \cite{imbrie:2014}. Interesting in this context is also the recent \cite{Mas} which outlines a mathematical proof of exponentially decaying correlations for the ground state of the XXZ chain in a quasi-periodic field of Harper-type, thus giving a (stationary) indication of MBL at low energy.

In the XY chain in a random field, the conserved quantities are given by fermionic eigenmodes that arise in the exact solution based on the Jordan-Wigner transformation \cite{LSM}. In the random case, these eigenmodes are obtained by diagonalizing an associated one-particle Anderson model on the one-dimensional lattice, and hence have strong localization properties \cite{Stolz}. They are not strictly local, however, because of the non-local character of the Jordan-Wigner transformations that maps the spin chain onto a system of quasi-free fermions. It has been shown previously that, as a consequence, the
XY chain in a random magnetic field shows complete dynamical localization. In fact, a strong form of dynamical localization,
i.e a zero-velocity Lieb-Robinson bound on-average, was proven in \cite{HamzaSimsStolz}. An averaged bound on dynamic correlations,
which is uniform in all eigenstates, is proven in \cite{SimsWarzel}.   Also, it was proved that the bipartite entanglement entropy of its ground states and energy eigenstates 
satisfy an area law \cite{PasturSlavin, AR-S}. In this work we prove that, for a large class of initial states, the entanglement entropy remains bounded for all times. Although the system could be regarded as integrable, the ballistic motion of quasi-particles responsible for the predicted ballistic growth of entanglement in homogeneous integrable systems is missing due to the localization phenomenon, see \cite{kim:2013} and references therein.

In the symmetric XY chain, sometimes referred to as the XX chain, with a random field in the $Z$ direction, the third component of the spin is conserved. In the fermion language this corresponds to the conservation of particle number. In this case one can obtain a good intuitive picture of the localization of the dynamics, and this is our first result: with overwhelming probability particles propagate only over a finite distance.
This is in stark contrast with the homogeneous XY chain \cite{antal:1998,antal:1999,ogata:2002}.
 In the general, asymmetric XY chain, the only conserved quantities are the occupation numbers of the quasi-particles given by the fermionic eigenmodes which, under our assumptions, are all exponentially localized. This implies that the relationship between the on-site fermions and the quasi-particles is itself quasi-local.

As mentioned above it has been argued that, generically, this should imply that the bipartite entanglement between two halves of the chain should grow no faster than logarithmically. For the class of initial states we consider here, which includes arbitrary product vectors, we prove a uniform bound on the entanglement. It is still possible that for some separable initial states the bipartite entanglement increases logarithmically. In any case, this is
a very small amount of entanglement compared to a random state of the chain which has entropy of entanglement approaching its maximum
possible value with near full
probability \cite{hayden:2006}. In view of the probability-zero nature of the states with bounded or logarithmic entanglement, one may wonder whether
the asymptotic behavior for large time is reliably accessible through numerical computation. It remains to be seen whether logarithmic entanglement
growth is a robust indicator of MBL. A recent proposal to detect MBL experimentally by measuring the entanglement is given in \cite{ho:2015}.

\subsection{The Model and Main Results}\label{sec:main results}

For any $n \geq 1$, we consider an anisotropic XY spin chain in transversal magnetic field on $[1,n] := \{1,\ldots,n\}$, given by the self-adjoint Hamiltonian
\begin{equation} \label{eq:anisoxychain}
H = H_{[1,n]} = - \sum_{j=1}^{n-1} \mu_j [ (1+\gamma_j) \sigma_j^x \sigma_{j+1}^x + (1-\gamma_j) \sigma_j^y \sigma_{j+1}^y] - \sum_{j=1}^n \nu_j \sigma_j^z \, 
\end{equation}
acting on the Hilbert space $\mathcal{H} = \bigotimes_{j=1}^n \C^2$. The parameters $\mu_j$, $\gamma_j$, and $\nu_j$ describe the interaction strength, anisotropy and field strength, respectively.
By $\sigma_j^x$, $\sigma_j^y$, and $\sigma_j^z$ we denote the standard Pauli matrices acting on the $j$-th component of the tensor product
$\mathcal{H}$. Our main interest is in the behavior of random systems, and so we often think of
the parameters, indicated above, as the first $n$ components of sequences of real-valued random
variables indexed by $j\in \N$. To be more precise, our standing assumptions will be that all three
sequences $\{ \mu_j \}_{j \in \N}$, $\{ \gamma_j \}_{j \in \N}$, and $\{ \nu_j \}_{j \in \N}$
are i.i.d., also independent from one another, and that they have distributions of bounded support.

More assumptions on the random parameters will be implicit, as we will state our results under the condition of {\it eigencorrelator localization} of the {\it effective one-particle Hamiltonian} $M$ associated with $H$, given by (\ref{effHam}) below. This will be understood as the existence of
a non-increasing function $F:[0,\infty)\to (0,\infty)$, of which we will require that it vanishes sufficiently fast as the argument tends to $\infty$, and such that
\begin{equation} \label{ecorloc}
\E \left( \sup_{|g|\le 1} \|g(M)_{jk}\| \right) \le F(|j-k|),
\end{equation}
uniformly in $n\in \N$ and $1\le j,k \le n$. Here, the supremum is taken over arbitrary Borel functions $g:\R \to \C$ with modulus pointwise bounded by $1$ and $g(M)$ is defined via the functional calculus of hermitean matrices.
Just as $M$ in (\ref{effHam}), we view $g(M)$ as an $n\times n$-matrix with $2\times 2$-matrix-valued entries and thus $\|\cdot\|$ on the left of (\ref{ecorloc}) is a norm on the $2\times 2$-matrices, which, for definiteness, we choose to be the Euclidean matrix norm.

Typical examples of the function $F$ in the RHS of (\ref{ecorloc}) are $F(r)=Ce^{-r/\xi}$, and $F(r)=C/(1+r)^{\beta}$, where $\xi$, $\beta$ and $C$
are positive constants.
More specific assumptions, e.g., on the size of $\beta$, as required for our results, will be given later. At the end of this introduction we will discuss conditions on the random variables for which (\ref{ecorloc}) is known for suitable $F$.

We have two main types of results. Both describe the behavior of certain dynamic quantities in the
regime of many-body localization. By assuming eigencorrelator localization (\ref{ecorloc}), our results show that one-partice localization of the effective Hamiltonian $M$ implies certain forms of many-body localization for the disordered XY spin chain $H$.

Our first result concerns `particle number transport' and is restricted to the case of the isotropic XY chain, i.e.\ we assume that $\gamma_j =0$ for all $j$. In this case the Hamiltonian $H$ commutes with the number operator
\begin{equation} \label{numaa*}
\mathcal{N} = \sum_{j=1}^n a_j^*a_j ,\quad \mbox{where} \quad  a^*=\begin{pmatrix} 0 & 1 \\ 0 & 0 \\ \end{pmatrix} \quad
\mbox{and} \quad  a=\begin{pmatrix} 0 & 0 \\ 1 & 0 \\ \end{pmatrix}
\end{equation}
are the basic raising and lowering operators. Denote the up-down-spin product basis states by
\begin{equation} \label{eq:updown}
e_{\alpha} := e_{\alpha_1} \otimes \ldots \otimes e_{\alpha_n}, \quad \alpha \in \{0,1\}^n,
\end{equation}
where
\begin{equation}
e_{0}:=\begin{pmatrix}
                    0 \\
                    1 \\
                  \end{pmatrix}=|\downarrow\rangle \ \text{and}\ \ e_{1}:=\begin{pmatrix}
                                                                       1 \\
                                                                       0 \\
                                                                     \end{pmatrix}=|\uparrow\rangle .
\end{equation}
The number operators counts the number of up-spins in $e_{\alpha}$: $\mathcal{N} e_\alpha = ke_\alpha$, where $k = \# \{j: \alpha_j=1\}$. The isotropic XY chain $H$ leaves the eigenspaces to the eigenvalues $k=0,\ldots,n$ of $\mathcal{N}$ invariant.

For a subset $S\subset \Lambda := [1,n]$ we also define the local number operator, measuring the number of up-spins in $S$, as
\begin{equation}
\mathcal{N}_S = \sum_{j \in S} a_j^*a_j.
\end{equation}

As initial states $\rho$ we will consider product states given by arbitrary density profiles,
\begin{equation} \label{densityprofile}
\rho = \bigotimes_{j=1}^n \rho_j, \quad \rho_j = \begin{pmatrix} \eta_j & 0 \\ 0 & 1-\eta_j \end{pmatrix}, \quad 0 \le \eta_j \le 1, \quad j=1,\ldots,n.
\end{equation}

For these initial states $\rho$ we are interested in how the expectation of the observable $\mathcal{N}_S$ changes under the Heisenberg evolution of $H$. The latter is given by
$\tau_t(A) = e^{itH}Ae^{-itH}$ for any $A \in \mathcal{B}( \mathcal{H})$, the bounded linear operators over $\mathcal{H}$, while the Schr\"odinger evolution of the state $\rho$ is $\rho_t = e^{-itH} \rho e^{itH}$. Denoting the expectation of an observable $A$ in the state $\rho$ by  $\langle A \rangle_{\rho} := {\rm Tr} [ A  \rho ]$, we thus have to analyze the quantity
\begin{equation} \label{pnfq}
\langle \mathcal{N}_{S} \rangle_{\rho_t} = \langle \tau_t( \mathcal{N}_{S}) \rangle_{\rho} .
\end{equation}

With $\mathbb{E}(\cdot)$ denoting the disorder average, we will prove the following bound on the number of up-spins (``particles'') in $S$:
\begin{thm}\label{thm:PNF}
Consider $H$ with $\gamma_j=0$ for all $j$ and assume eigencorrelator localization (\ref{ecorloc}) for the effective Hamiltonian $M$.  Then
\begin{equation} \label{eq:PNF}
\mathbb{E}\left(\sup_t\langle \mathcal{N}_{S} \rangle_{\rho_t}\right)\leq \sum_{j\in S} \sum_{k=1}^n \eta_k F(|j-k|).
\end{equation}
\end{thm}

With the same proof, see Section~\ref{sec:pnf} below, we can similarly bound the number of down-spins (``holes''): If $\tilde{\mathcal{N}}_S := \sum_{j\in S} a_j a_j^* = |S|\idty - \mathcal{N}_S$, then
\begin{equation} \label{eq:HNF}
\mathbb{E}\left(\sup_t\langle \tilde{\mathcal{N}}_{S} \rangle_{\rho_t}\right)\leq \sum_{j\in S} \sum_{k=1}^n (1-\eta_k) F(|j-k|).
\end{equation}

An interesting special case of Theorem~\ref{thm:PNF} arises when starting with all up-spins in a subinterval $\Lambda_0 = [a,b] \subset \Lambda$ and all downspins in $\Lambda \setminus \Lambda_0$, i.e.\ for the pure state $\rho = |\varphi\rangle \langle \varphi|$ with
\begin{equation} \label{updownvec}
\varphi = |\downarrow\rangle^{\otimes (a-1)}\otimes |\uparrow\rangle^{\otimes (b-a+1)}\otimes|\downarrow\rangle^{\otimes(n-b)},
\end{equation}
meaning $\eta_k =1$ for $k\in \Lambda_0$ and $\eta_k=0$ otherwise in (\ref{densityprofile}). If $S\subset \Lambda \setminus \Lambda_0$, then (\ref{eq:PNF}) implies
\begin{equation} \label{eq:PNbound}
\mathbb{E}\left(\sup_t\langle \mathcal{N}_{S} \rangle_{\rho_t}\right)\leq \sum_{j\in S} \sum_{k\in \Lambda_0} F(|j-k|) \leq 2 \sum_{j=d(S,\Lambda_0)}^\infty j F(|j|),
\end{equation}
where $d(A,B) = \min \{|a-b|: a\in A, b\in B\}$ denotes the distance between two sets. This is a bound on the expectation of the number of up-spins which penetrate from $\Lambda_0$ into $S$. If $F$ has sufficiently rapid decay, e.g.\ $F(r) \le C/(1+r)^{\beta}$ for some $\beta>2$, then the right hand side of (\ref{eq:PNbound}) is not only finite, uniformly in the sizes of $\Lambda$ and $\Lambda_0$, but decaying for growing distance $d(S,\Lambda_0)$. More precisely, for the case of power decay $F(r) = C/(1+r)^{\beta}$, we get
\begin{equation} \label{eq:noPT}
\mathbb{E}\left(\sup_t\langle \mathcal{N}_{S} \rangle_{\rho_t}\right)\leq C'/d(S,\Lambda_0)^{\beta-2}.
\end{equation}
Similarly, exponential decay of $F$ leads to exponential decay of $\mathbb{E}\left(\sup_t\langle \mathcal{N}_{S} \rangle_{\rho_t}\right)$ in $d(S,\Lambda_0)$ (with slightly reduced decay rate).

The proof of Theorem~\ref{thm:PNF} does not extend to the case of the anisotropic XY chain, as can be expected physically: In the anisotropic case we do not have particle number conservation. Thus up-spins can be created by local properties of the dynamics, for example by simply flipping a down-spin into an up-spin. This can not be prevented by being in the regime of many-body localization, which is the main physical mechanism exploited in our proof.

Results similar to Theorem~\ref{thm:PNF} have recently been shown in \cite{SW} for the disordered Tonks-Girardeau gas, which can be understood as a continuum analogue of the isotropic XY chain. In fact, the corresponding Proposition~2.1 of \cite{SW} mostly serves as a warm-up for deeper results on absence of Bose-Einstein condensation and superfluidity. As was pointed out to us by a referee, the methods used to prove Propostion~2.1 in \cite{SW} can easily be adjusted to our discrete setting and provide two nice extensions of Theorem~\ref{thm:PNF}. We will discuss this at the end of Section~\ref{sec:pnf}.

Our second result concerns the entanglement dynamics of states formed by taking certain products of eigenstates.
The set-up is as before, but now we consider the general anisotropic XY chain Hamiltonian $H=H_{[1,n]}$, as in (\ref{eq:anisoxychain}), on volume $\Lambda = [1,n]$. We will need to assume that, for every $n\in \N$,
\begin{equation} \label{eq:simple}
H_{[1,n]} \:\mbox{almost surely has simple spectrum}.
\end{equation}

Let $\Lambda_0 = [a,b]$ be an arbitrary subinterval of $\Lambda$ and consider the bi-partite decomposition
\begin{equation} \label{bipd}
\mathcal{H} = \mathcal{H}_1 \otimes \mathcal{H}_2 \quad \mbox{with} \quad \mathcal{H}_1 = \bigotimes_{j\in \Lambda_0} \C^2 \quad \mbox{and} \quad \mathcal{H}_2 = \bigotimes_{j\in \Lambda \setminus \Lambda_0} \C^2.
\end{equation}

By $\rho^1 = \Tr_{\mathcal{H}_2} \rho$ we denote the partial trace of a state $\rho$ in $\mathcal{H}$ found by tracing out the variables in $\Lambda \setminus \Lambda_0$ and $\mathcal{E}(\rho) = \mathcal{S}(\rho^1)$ the entanglement of $\rho$ with respect to the decomposition of $\Lambda$ into $\Lambda_0$ and $\Lambda \setminus \Lambda_0$.

We will consider the Schr\"odinger time evolution $\rho_t = e^{-itH} \rho e^{itH}$ of suitable initial states $\rho$ and study how their bi-partite entanglement $\mathcal{E}(\rho_t)$ with respect to this decomposition grows in time. Towards this goal we can handle initial states which are products of {\it any finite number} of eigenstates of restrictions of the XY Hamiltonian (\ref{eq:anisoxychain}) to subsystems.

More precisely, let
$1 = r_0 < r_1 < \cdots < r_m =n$ be integers and, for each $1 \leq k < m$,
set $\Lambda_k = [ r_{k-1}, r_k-1]$, while $\Lambda_m = [r_{m-1}, r_m]$.
Thus $\Lambda$ is a disjoint union of the intervals $\Lambda_k$ with $k=1, \cdots, m$.

For $1 \leq k \leq m$, consider the restrictions $H_{\Lambda_k}$ of the XY Hamiltonian $H=H_{\Lambda}$ to $\Lambda_k$, defined similar to (\ref{eq:anisoxychain}), which are self-adjoint operators on $\mathcal{H}_{\Lambda_k} = \bigotimes_{j \in \Lambda_k} \mathbb{C}^2$.

For each $k$ let $\psi_k$ be a normalized eigenstate of $H_{\Lambda_k}$ and let $\rho_{\psi_k} = |\psi_k \rangle \langle \psi_k|$. We choose the initial state
\begin{equation} \label{eq:instate}
\rho = \bigotimes_{k=1}^m \rho_{\psi_k}.
\end{equation}

We can now state our second main result, an area law for the entanglement dynamics of product states of the form (\ref{eq:instate}):

\begin{thm}\label{thm:DAL}
Assume that the anisotropic random XY chain (\ref{eq:anisoxychain}) has almost sure simple spectrum (\ref{eq:simple}) and that $M$ satisfies eigencorrelator localization (\ref{ecorloc}) with $F(r) = C/(1+r)^{\beta}$ for some $\beta>6$. Consider an initial state $\rho$ as given by (\ref{eq:instate}) and its Schr\"odinger evolution $\rho_t = e^{-iHt} \rho e^{iHt}$ under the full XY chain Hamiltonian $H$.

Then there exists $C<\infty$ such that
\begin{equation}\label{AreaLaw}
\mathbb{E}\left(\sup_{t,\{\psi_k\}_{k=1,2,\ldots,m}}\mathcal{E}(\rho_t)\right) \le C
\end{equation}
for all $n$, $m$, $a$ and $b$ with $1\le a \le b \le n$, $1\le m \le n$ and all decompositions $\Lambda_1,\ldots, \Lambda_m$ of $\Lambda=[1,n]$. In (\ref{AreaLaw})  the supremum is taken over all $t\in \R$ and all normalized eigenfunctions $\psi_k$ of $H_{\Lambda_k}$, $k=1,\ldots,m$.
\end{thm}

In general, we do not require the decomposition $\Lambda_1, \ldots, \Lambda_m$ to be compatible with the decomposition into $\Lambda_0$ and $\Lambda \setminus \Lambda_0$. However, if $\Lambda_0$ is chosen to be a union of adjacent $\Lambda_k$, say of $\Lambda_r, \ldots, \Lambda_s$, then the initial state is unentangled with respect to $\mathcal{H}_1 \otimes \mathcal{H}_2$:
\begin{equation}
\mathcal{E}(\rho_{t=0}) = \mathcal{E}(\rho) = \mathcal{S}(\rho_{\psi_r} \otimes \ldots \otimes \rho_{\psi_s}) = 0,
\end{equation}
as $\rho_{\psi_r} \otimes \ldots \otimes \rho_{\psi_s}$ is a pure state. For $t>0$, $\rho_t$ is generally not a product state and thus $(\rho_t)^1$ not a pure state, so that $\mathcal{E}(\rho_t)$ will be strictly positive (with upper bound given by the `volume law' $\log \dim \mathcal{H}_1 = |\Lambda_0| \log 2$).

The `area law' (\ref{AreaLaw}), giving an upper bound for the entanglement dynamics proportional to the surface area of the subsystem $\Lambda_0$ (with surface given by its two endpoints), is uniform not only in time $t\in \R$, the size of the system $\Lambda$, and the subsystem $\Lambda_0$, but also applies uniformly to {\it all} possible products of eigenstates of $H_{\Lambda_k}$ for $k=1,2,\ldots,m$, irrespective of their energy. As a special case one could choose a product of ground states of the $H_{\Lambda_k}$, but our result goes far beyond this and reflects the fact that the random XY chain is a model of a {\it fully many-body localized} quantum system.

It remains an open question if a uniform bound such as (\ref{AreaLaw}) on the dynamical entanglement holds also if general pure states $\psi_k$ in the subsystems $\mathcal{H}_{\Lambda_k}$ are used in the initial condition (\ref{eq:instate}), instead of eigenstates of $H_{\Lambda_k}$. Alternatively, such initial conditions might realize a slow temporal growth of the entanglement, at least up to saturation at the maximal entropy $|\Lambda_0| \log 2$ in $\Lambda_0$, as predicted in the physics literature to be possible in the MBL regime. At this point we do not have a good conjecture on what to expect.

We comment on the two extreme cases $m=n$ and $m=1$:

(i) In the extreme case when $m=n$ we have $\Lambda_k = \{k\}$ and $H_{\Lambda_k}=-\nu_{k}\sigma^Z_{k}$, so that all interaction terms in the XY chain have been removed.
For each $k$ we have $\sigma(H_{\Lambda_k})=\{\nu_{k},-\nu_{k}\}$, which is almost surely simple if the distribution of the $\nu_k$ does not have an atom at $0$.
The eigenvectors are $|\uparrow\rangle$ and $|\downarrow\rangle$, so that the corresponding initial states for Theorem~\ref{thm:DAL} become $\rho=|e_{\alpha}\rangle \langle e_{\alpha}|$ with arbitrary up-down-spin configurations $e_{\alpha}$ given by (\ref{eq:updown}).

Then Theorem \ref{thm:DAL} gives that eigencorrelator localization of the effective Hamiltonian $M$ implies an area law for the Schr\"odinger evolution $e^{-itH} |e_\alpha \rangle \langle e_\alpha| e^{itH}$ of arbitrary up-down-spin configurations. In fact, in this case the proof allows to slightly weaken the required decaying rate of the eigencorrelators:
\begin{cor}\label{cor:updown}
Assume that $H$ is almost surely simple and that $M$ has localized eigencorrelators with $F(r) = C/(1+r)^{\beta}$ for some $\beta>4$.

Then there exists $C<\infty$ such that
\begin{equation} \label{eq:arealaw}
\E \left( \sup_{t,\alpha} \mathcal{E}(e^{-itH} |e_\alpha \rangle \langle e_\alpha| e^{itH}) \right) \le C
\end{equation}
for all $n$, $a$ and $b$ with $1\le a \le b \le n$.
\end{cor}

We comment at the end of Section~\ref{sec:DAL} how the proof of Theorem~\ref{thm:DAL} simplifies in this case, so that $\beta>4$ suffices.

(ii) In the other extreme case when $m=1$. In this case (\ref{eq:instate}) simply means that $\rho=\rho_{\psi}$ for a normalized eigenstate $\psi$ of $H$. As these states are stationary under the time-evolution, $\rho_t = \rho$, Theorem~\ref{thm:DAL}, respectively its proof, yields the main result of \cite{AR-S} as a special case:

\begin{cor}[\cite{AR-S}] \label{cor:AR-S}
Assume that $H$ is almost surely simple and that $M$ has localized eigencorrelators with $F(r) = C/(1+r)^{\beta}$ for some $\beta>2$. There exists $C<\infty$ such that
\begin{equation}
\mathbb{E}\left(\sup_{\psi}\mathcal{E}(|\psi\rangle \langle\psi|)\right) \le C,
\end{equation}
 where the supremum is taken over all eigenvectors $\psi$ of $H$.
 \end{cor}

Again, we will comment on the proof of this at the end of Section~\ref{sec:DAL}.

Let us briefly discuss the validity of our main assumption, i.e.\ eigencorrelator localization (\ref{ecorloc}), in the above theorems. Here we consider the best understood case of random magnetic field $\{\nu_j\}$ and constant parameters $\mu = \mu_j$, $\gamma = \gamma_j$. Eigencorrelator localization with exponential decay, i.e.\ $F(r) = Ce^{-r/\xi}$, is known for the following cases: (i) The isotropic case $\gamma =0$, where $M$ reduces to the classical Anderson model (see, e.g., the survey \cite{Stolz}). (ii) The general anisotropic model at large disorder, i.e. $\nu_j = \lambda \tilde{\nu}_j$, where $\lambda$ is sufficiently large and $\tilde{\nu}_j$ i.i.d.\ with bounded compactly supported density \cite{Elgartetal}.

Finally, eigencorrelator localization with subexponential decay $F(r) = Ce^{-\eta r^{\xi}}$, $\xi<1$, has been shown in \cite{ChapmanStolz} for the general anisotropic case under the assumption of a zero-energy spectral gap for $M$, i.e.\ the existence of $C> |\mu|$ such that $\nu_j \ge C$ for all $j$ or $\nu_j \le -C$ for all $j$. This only requires non-triviality of the i.i.d.\ random variables $\nu_j$ (i.e.\ their distribution is supported on more than one point). This suffices for Theorem~\ref{thm:PNF}. However, in Theorem~\ref{thm:DAL} we also need almost sure simplicity of the spectrum of $H$. This is shown in Appendix A of \cite{AR-S} under the assumption that the $\eta_j$ have absolutely continuous distribution.

The rest of this paper is organized as follows. In Section~\ref{sec:effham}, we briefly recall some of the basics
concerning the diagonalization of the XY spin chain. Section~\ref{sec:pnf} contains the short proof of Theorem~\ref{thm:PNF} on particle number transport. The proof of Theorem~\ref{thm:DAL} in Section~\ref{sec:DAL} is more involved and needs additional tools.  The general framework for this is the theory of quasi-free states in free Fermion systems and how they can be studied in terms of correlation matrices. In particular, we recall the formula (\ref{eq:EE:Wick}), expressing the entanglement of states in terms of restrictions of correlation matrices. This goes back to \cite{VLRK} and was proven in \cite{AR-S} in the generality required here. Some more general facts about the CAR algebra and quasi-free states which are used in this article are collected in Section~\ref{app:carqf}. This includes a proof of the general fact that products of quasi-free states are quasi-free.

%
%
%

\section{On the Effective Single-Particle Hamiltonian}\label{sec:effham}

One of the most important features of the XY spin chain is that the many-body Hamiltonian, which acts on
a Hilbert space of dimension $2^n$, can be reduced to an effective single-particle operator that acts on a space
with dimension linear in $n$. We briefly introduce some notation associated with this well-known result.

%

Consider the XY Hamiltonian $H$, as in (\ref{eq:anisoxychain}), on the interval $[1,n]$.
In terms of the lowering operator, see (\ref{numaa*}), and the Pauli matrix $\sigma^z$ define
a collection of $c$-operators by setting
\begin{equation} \label{eq:cjs}
c_1 := a_1 \quad \mbox{and} \quad  c_j := \sigma_1^z \ldots \sigma_{j-1}^z a_j \quad \mbox{for } j= 2,\ldots,n.
\end{equation}
This {\it change of variables} is known as the Jordan-Wigner transform.
These non-local operators are easily checked to satisfy the canonical anti-commutation relations (CAR):
\begin{equation} \label{CAR}
\{ c_j, c_k \} = \{ c_j^*, c_k^* \} = 0 \quad \mbox{and} \quad \{c_j, c_k^* \} = \delta_{j,k} \idty \quad \mbox{for all } 1 \leq j,k \leq n.
\end{equation}
Crucial for the theory of the XY Hamiltonian $H$ is that it can be expressed as a quadratic form in the operators $c_j$ and their adjoints. The coefficient matrices appearing in these representations provide effective one-particle Hamiltonians for $H$. Below we present these representations separately for the isotropic and anisotropic case. The former is simpler and most suitable for proving Theorem~\ref{thm:PNF} in Section~\ref{sec:pnf}, while the latter is needed for the proof of Theorem~\ref{thm:DAL} in Section~\ref{sec:DAL}.

\subsection{Isotropic Case}
Let $H_{\text{iso}}$ denote the isotropic XY chain, i.e.\ the special case of $H$ with $\gamma_j=0$ for all $j$.
One finds that the Hamiltonian $H_{\text{iso}}$ can be rewritten as
\begin{equation}\label{eq:isotropicH}
H_{\text{iso}}=2c^* A c + E_0 \I.
\end{equation}
Here $c := (c_1,\ldots, c_n)^T$, $c^* = (c_1^*, \ldots, c_n^*)$,  $E_0 := \sum_j \nu_j$ and $A$ is the symmetric Jacobi matrix
\begin{equation} \label{eq:A}
A= (A_{jk})_{j,k=1}^n :=  \left( \begin{matrix} -\nu_1 & \mu_1 & & \\ \mu_1 & \ddots & \ddots & \\ & \ddots & \ddots & \mu_{n-1} \\ & & \mu_{n-1} & -\nu_n \end{matrix} \right).
\end{equation}
There is an orthogonal matrix $U$ such that
\begin{equation} \label{eq:diagA}
U A U^t = \tilde{\Lambda} = \mbox{diag}(\tilde{\lambda}_j),
\end{equation}
where $\tilde{\lambda}_j$, $j=1,\ldots,n$, are the eigenvalues of $A$. Define $\tilde{b}=(\tilde{b}_1,\ldots, \tilde{b}_n)^T$ by
\begin{equation} \label{eq:bdef}
\tilde{b}= Uc.
\end{equation}
By the bi-linearity of $\{\cdot,\cdot\}$ and orthogonality of $U$ it follows readily that $\tilde{b}_j$, $j=1,\ldots,n$ satisfies the CAR. Moreover,
\begin{eqnarray} \label{eq:Fermi:Iso}
H_{\text{iso}} & = & 2c^* U^t \tilde{\Lambda} U c +E_0 \I  = 2 \tilde{b}^* \tilde{\Lambda} \tilde{b} +E_0 \I \nonumber \\ & = & 2\sum_{j=1}^n \tilde{\lambda}_j \tilde{b}_j^* \tilde{b}_j + E_0 \I.
\end{eqnarray}
Thus $H_{\text{iso}}$ has been written in the form of a {\it free Fermion system}.

\subsection{Anisotropic Case}
The XY Hamiltonian in its general anisotropic form (\ref{eq:anisoxychain}) can be written as
\begin{equation} \label{eq:Crep}
H =   \mathcal{C}^* M \mathcal{C},
\end{equation}
where we now use a vector notation for the collections
\begin{equation} \label{JWsystem}
\mathcal{C}= (c_1, c_1^*,c_2,c_2^*, \ldots, c_n, c_n^*)^T \quad \mbox{and similarly} \quad \mathcal{C}^*= (c_1^*,c_1  \ldots, c_n^*, c_n).
\end{equation}
The $2n\times 2n$ coefficient matrix
\begin{equation} \label{effHam}
M := \left( \begin{array}{cccc} -\nu_1 \sigma^z & \mu_1 S(\gamma_1) & & \\ \mu_1 S(\gamma_1)^T & -\nu_2 \sigma^z & \ddots & \\ & \ddots & \ddots & \mu_{n-1} S(\gamma_{n-1}) \\ & & \mu_{n-1} S(\gamma_{n-1})^T & -\nu_{n} \sigma^z \end{array} \right),
\end{equation}
where
\begin{equation}
S(\gamma) = \begin{pmatrix} 1 & \gamma \\ -\gamma & -1 \end{pmatrix}.
\end{equation}
It is well know that $M$ is diagonalizable by a (real-valued) Bogoliubov matrix $W$,
\begin{equation} \label{eq:Mdiag}
W M W^t=\bigoplus_{j=1}^n\begin{pmatrix}
                           \lambda_j & 0 \\
                           0 & -\lambda_j \\
                         \end{pmatrix},
\end{equation}
where $0\leq\lambda_1\leq\lambda_2\leq\ldots\leq\lambda_n$.
The Bogoliubov matrix $W$ defines the Fermionic system
\begin{equation}\label{cov}
\mathcal{B}:=W\mathcal{C}=(b_1,b_1^*,b_2,b_2^*,\ldots,b_n,b_n^*)^T,
\end{equation}
and
\begin{eqnarray} \label{freefermi}
H & = & \mathcal{C}^* M \mathcal{C} = \mathcal{B}^* W M W^t \mathcal{B} = \mathcal{B}^* \bigoplus_{j=1}^n\left( \begin{array}{cc} \lambda_j & 0 \\ 0 & -\lambda_j \end{array} \right) \mathcal{B}  \\
& = & \sum_{j=1}^n \lambda_j (b_j^* b_j - b_j b_j^*)
 =  2 \sum_{j=1}^n \lambda_j b_j^* b_j - E_1 \I. \nonumber
\end{eqnarray}
where $E_1:=\sum_{j=1}^n \lambda_j$. Thus, as in the isotropic case treated above, the Hamiltonian has been written as a free Fermion system. We note that the main difference between the isotropic and anisotropic cases is that, as opposed to (\ref{eq:bdef}), the Bogoliubov transform (\ref{cov}) generally `mixes' the $c$ and $c^*$ modes.

Up to a phase, the free Fermion system has a unique normalized vacuum vector $\Omega \in \bigcap_j \mbox{ker}\,b_j$,
and the collection
\begin{equation}\label{eq:PsiAlpha}
\psi_{\alpha} = (b_1^*)^{\alpha_1} \ldots (b_n^*)^{\alpha_n} \Omega, \quad \alpha =(\alpha_1,\alpha_2,\ldots,\alpha_n) \in \{0,1\}^n,
\end{equation}
forms an orthonormal basis (ONB) of the Hilbert space $\mathcal{H} = \bigotimes_{j=1}^n \C^2$.
In fact, the $\psi_{\alpha}$ form a complete set of eigenvectors for $H$ with
corresponding eigenvalues: $E_{\alpha} = 2\sum_{j:\alpha_j=1} \lambda_j - \sum_{j=1}^n \lambda_j$.

We may also consider the $c$-operators as associated with a free Fermion system. In this case, the vacuum vector in the kernels of all $c_j$ is $e_0$, the all-spins down vector, and the basis $(c_1^*)^{\alpha_1} \ldots (c_n^*)^{\alpha_n} e_0$, $\alpha \in \{0,1\}^n$, consists of the up-down spin product vectors $e_{\alpha}$ (up to a sign). It is in this sense that up-spins are interpreted as particles associated with the Fermionic creation operators $c_j^*$.

%
%

\section{Particle Number Transport} \label{sec:pnf}

In this section, we prove Theorem \ref{thm:PNF}.
The quantity of interest here is
\begin{equation} \label{tenum}
\langle \mathcal{N}_{S} \rangle_{\rho_t} = \sum_{j \in S} \langle a_j^* a_j \rangle_{\rho_t}
= \sum_{j \in S} \langle c_j^* c_j \rangle_{\rho_t}=\sum_{j \in S} \langle \tau_t(c_j^* c_j) \rangle_{\rho}=\sum_{j \in S} \langle \tau_t(c_j^*) \tau_t(c_j) \rangle_{\rho}
\end{equation}
where we have used that this product of raising and lowering operators is equal to the identical product of the
corresponding Jordan-Wigner variables, see (\ref{eq:cjs}), i.e., $a_j^*a_j = c_j^* c_j$ for all $1 \leq j \leq n$.

Since the Hamiltonian generating the dynamics $H_{\text{iso}}$ has a particularly simple form in terms of the $\tilde{b}$-operators, see (\ref{eq:Fermi:Iso}), one finds that
\begin{equation}
\tau_t(\tilde{b}_j) = e^{-2it \tilde{\lambda}_j} \tilde{b}_j.
\end{equation}
In vector form, this can be expressed as $\tau_t(\tilde{b})=e^{-2it\tilde{\Lambda}}\tilde{b}$. This implies
\begin{equation}
 \tau_t(c)=e^{-2itA}c
\end{equation}
as a consequence of the change of variables (\ref{eq:bdef}). Thus we get
\begin{equation}\label{eq:PN:Edynmcs}
\langle \tau_t(c_j^*) \tau_t(c_j) \rangle_{\rho}=\sum_{k,\ell=1}^n \overline{\left(e^{-2itA}\right)_{j\ell}}\left(e^{-2itA}\right)_{jk}\langle c_\ell^* c_k\rangle_{\rho}.
\end{equation}
It is easy to see, with the specific initial density matrix $\rho = \bigotimes_j \rho_j$ given in (\ref{densityprofile}), that
\begin{equation}\label{eq:PN:ccrho}
\langle c_\ell^* c_k\rangle_{\rho}=\delta_{\ell,k}\,\eta_k.
\end{equation}
By substituting (\ref{eq:PN:ccrho}) in (\ref{eq:PN:Edynmcs}) and then in (\ref{tenum}), we find
\begin{equation} \label{eq:PNS}
\langle \mathcal{N}_{S} \rangle_{\rho_t}=\sum_{j\in S}\sum_{k=1}^{n}\eta_k \left|(e^{2itA})_{jk}\right|^2.
\end{equation}
Thus we can use eigencorrelator localization (\ref{ecorloc}) to find that
\begin{eqnarray}
\mathbb{E} \left( \sup_t \langle \mathcal{N}_{S} \rangle_{\rho_t}  \right) & \leq &  \sum_{j\in S} \sum_{k=1}^n \eta_k \,\E \left(\sup_t |(e^{2itA})_{j,k}| \right) \\
& \le & \sum_{j\in S} \sum_{k=1}^n \eta_k F(|j-k|), \nonumber
\end{eqnarray}
as claimed.

We now comment on extensions of the above result which provide analogues to corresponding bounds for the disordered Tonks-Girardeau gas, see Proposition~2.1 in \cite{SW}. First one can show, with exactly the same notations as in Theorem~\ref{thm:PNF} above, that
\begin{equation} \label{eq:SWbound1}
\E \left( \sup_t | \langle \mathcal{N}_S \rangle_{\rho_t} - \langle \mathcal{N}_S \rangle_{\rho}| \right) \le 2 \sum_{j\in S, k\in S^c} F(|j-k|).
\end{equation}
To see this, we mimic notation used in \cite{SW} to write the expression found in (\ref{eq:PNS}) as
\begin{equation}
\left| \langle \mathcal{N}_{S} \rangle_{\rho_t} - \langle \mathcal{N}_S \rangle_{\rho} \right| = \left| \tr \chi_S U_t^* \eta U_t - \tr \chi_S \eta \right|,
\end{equation}
where $\eta = \mbox{\rm{diag}}(\eta_k)$ and $U_t = \exp(2itA)$. Further retracing \cite{SW} the latter is equal to
\begin{eqnarray}
\left| \tr \chi_S U_t^* \chi_{S^c} \eta U_t - \tr \chi_{S^c} U_t^* \chi_S \eta U_t \right| & \le & \| \chi_S U_t^* \chi_{S^c}\|_1 \|\eta U_t\| + \|\chi_{S^c} U_t^* \chi_S\|_1 \|\eta U_t\| \\
& \le & \sum_{j\in S, \,k\in S^c} |(U_t^*)_{jk}| + \sum_{j\in S^c,\, k\in S} |(U_t^*)_{jk}|. \nonumber
\end{eqnarray}
This yields (\ref{eq:SWbound1}). Similarly, and omitting details, one adapts Propositon~2.1.2 of \cite{SW} to get
\begin{equation} \label{eq:SWbound2}
\E \left( \sup_t \langle \mathcal{N}_S \rangle_{\rho_t} \right) \le \langle \mathcal{N}_K \rangle_{\rho} + \sum_{j\in S,\, k\in K^c} F(|j-k|)
\end{equation}
for arbitrary pairs of subsets $S \subset K \subset \Lambda$. In the case of sufficiently decaying $F$ and, say, intervals $S$ and $K$, (\ref{eq:SWbound1}) and (\ref{eq:SWbound2}) can be interpreted as absence/smallness of particle transport for general density profiles $\rho$ as in (\ref{densityprofile}). The application (\ref{eq:noPT}) of Theorem~\ref{thm:PNF}, on the other hand, required to assume that $S$ is initially void of particles.

%
%

\section{The Dynamical Entanglement Entropy of a Product of Eigenstates} \label{sec:DAL}

%
%

In this section we prove our second main result, Theorem~\ref{thm:DAL}. Our goal is to analyze the dynamical entanglement entropy $\mathcal{E}( \rho_t)$ of
 the state with initial density matrix $\rho = \bigotimes_{k=1}^m \rho_{\psi_k}$.

 We start by noting that the simplicity assumption (\ref{eq:simple}) requires that, almost surely, $H_{[1,n]}$ is simple for {\it all} $n\in \N$. Since we have also assumed that the random parameters entering $H$ are i.i.d., this yields by translation invariance that all $H_{\Lambda_k}$, $k=1,\ldots,m$, are simple, again almost surely.

 In Section~\ref{app:carqf} we recall that, under the assumption that $H$ has simple spectrum, all density matrices $\rho_{\psi} = |\psi \rangle \langle \psi|$ corresponding to eigenstates $\psi$ of $H$ are quasi-free states on $\mathcal{B}(\mathcal{H})$ (also see Section~\ref{app:carqf} for a definition). The same argument applies to restrictions of the XY chain to the subintervals $\Lambda_k$ of $\Lambda$.
 Thus, almost surely, the density matrices $\rho_{\psi_k}$, corresponding to eigenstates of $H_{\Lambda_k}$, define quasi-free states on $\mathcal {B}(\mathcal{H}_{\Lambda_k})$ for all $k=1,\ldots,m$. Here, on a subinterval $\Lambda_k$ one defines quasi-free with respect to `local' Jordan-Wigner operators $c_j^{(k)}$, $j \in \Lambda_k$, given similar to (\ref{eq:cjs}). We will later refer to
 \begin{equation}
 \mathcal{C}^{(k)} := (c_1^{(k)},(c_1^{(k)})^*, \ldots, c_{|\Lambda_k|}^{(k)}, (c_{|\Lambda_k|}^{(k)})^*)
 \end{equation}
  as the {\it local Jordan-Wigner system} on $\Lambda_k$.

 By iterated application of Lemma~\ref{lem:addqf}, we conclude that
 $\rho$ is a density matrix corresponding to a quasi-free state on $\mathcal{B}(\mathcal{H})$. Finally, by Lemma~\ref{lem:qftime}, we find that $\rho_t$ is quasi-free for all times.
 Hence we are analyzing the entanglement entropy of
 a quasi-free state.

 The importance of this lies in the fact that quasi-free states $\rho$ are completely characterized by their correlation matrices $\Gamma_{\rho}^{\mathcal C}$ defined as
 \begin{equation} \label{eq:cormatrix}
\Gamma_{\rho}^{\mathcal{C}} : = \langle \mathcal{C} \mathcal{C}^* \rangle_{\rho}
\end{equation}
i.e.\ the $2n \times 2n$ complex matrix with entries $( \Gamma_{\rho}^{\mathcal{C}})_{j,k} = \langle (\mathcal{C} \mathcal{C}^*)_{j,k} \rangle_{\rho}$. Moreover, the von Neumann entropy of a reduced state $\rho^1$ is determined by a restriction of the correlation matrix.

\begin{lem}{\cite{AR-S}} \label{lem:AR-S}
Let $\rho$ be a quasi-free state on $\mathcal{B}(\mathcal{H})$. The entanglement entropy of $\rho$ with respect to the decomposition $\mathcal{H}_1\otimes \mathcal{H}_2$ in (\ref{bipd}) is given by the formula:
\begin{equation}\label{eq:EE:Wick}
  \mathcal{E}(\rho) = -\Tr \rho^1 \log \rho^1 = - \tr \Gamma_{\rho^1}^{\mathcal{C}_1} \log \Gamma_{\rho^1}^{\mathcal{C}_1}
\end{equation}
where $\mathcal{C}_1$ is the local Jordan-Wigner system on $\Lambda_0$. Moreover, the correlation matrix $\Gamma_{\rho^1}^{\mathcal{C}_1}$ of $\rho^1$ is the restriction of $\Gamma_\rho^{\mathcal{C}}$ to $\Span\{e_{2j-1},e_{2j},j\in\Lambda_0\}$.
\end{lem}

We note that, for the case of the ground state of the XY chain in constant magnetic field, the identity (\ref{eq:EE:Wick}) was first given in \cite{VLRK}. There, together with an exact expression for the correlation matrix in the thermodynamic limit $\Lambda \to \Z$, (\ref{eq:EE:Wick}) was used to numerically predict the dependence of the ground state entanglement on the size of the subsystem. These predictions, again for constant field and in the thermodynamic limit, were later rigorously proven, see \cite{IJK, Itsetal} for most general results and additional references.

Applying Lemma~\ref{lem:AR-S} to $\rho_t$, we see that we must investigate the correlation matrix $\Gamma_{\rho_t}^{\mathcal{C}}$.
To see how this correlation matrix evolves in time, we use the 
following lemma.
\begin{lem} \label{lem:corr:tautC}
Let $\rho$ be a density matrix on $\mathcal{H}$. Then
\begin{equation} \label{timecor}
\Gamma_{\rho_t}^{\mathcal{C}}=e^{-2itM}\Gamma^{\mathcal{C}}_{\rho}e^{2itM}.
\end{equation}
\end{lem}
\begin{proof}
From the free Fermion form (in terms of the $b$-operators) of the Hamiltonian generating the dynamics, see (\ref{freefermi}),
one finds that
\begin{equation}
\tau_t(b_j) = e^{-2it \lambda_j} b_j \quad \mbox{and} \quad \tau_t(b_j^*) = e^{2it \lambda_j} b_j^*.
\end{equation}
In vector form this can be expressed as
\begin{equation} \label{BCdynamics}
\tau_t( \mathcal{B}) = \bigoplus_{j=1}^n\begin{pmatrix}
                         e^{-2it \lambda_j} & 0 \\
                         0 & e^{2it \lambda_j} \\
                       \end{pmatrix}
\mathcal{B}, \quad \mbox{which then implies} \quad \tau_t( \mathcal{C}) = e^{-2it M} \mathcal{C},
\end{equation}
the latter formula a consequence of (\ref{eq:Mdiag}) and the change of variables (\ref{cov}). (For more details regarding
such calculations see e.g. Section 3.1 of \cite{HamzaSimsStolz}.)
As a result,
\begin{equation}
\Gamma_{\rho_t}^{\mathcal{C}} = \langle \tau_t(\mathcal{C} \mathcal{C}^*) \rangle_{\rho} =  \langle e^{-2it M} \mathcal{C} \mathcal{C}^* e^{2it M} \rangle_{\rho}
\end{equation}
and then (\ref{timecor}) follows using that $\langle \cdot \rangle_{\rho}$ is linear.
\end{proof}

Finally, to exploit (\ref{timecor}) we need to compute the correlation matrix of the initial density matrix $\rho$.
That is the content of the next lemma.

\begin{lem}\label{lem:corinirho}
The correlation matrix of $\rho=\bigotimes_{k=1}^m\rho_{\psi_k}$ with respect to $\mathcal{C}$ is
\begin{equation}\label{directsum}
\Gamma_{\rho}^{\mathcal{C}}=\bigoplus_{j=1}^m \Gamma_{\rho_{\psi_k}}^{\mathcal{C}^{(k)}}.
\end{equation}
\end{lem}

\begin{proof}
(\ref{directsum}) follows directly from the following argument, for $1 \leq j \leq j' \leq n$, then
\begin{equation}
 \langle (c_j)^{\#} (c_{j'})^{\#} \rangle_{\rho} = \left\{ \begin{array}{cc} 0 & \mbox{if } j \in \Lambda_k \mbox{ and } j' \in \Lambda_{k'} \mbox{ with } k \neq k', \\
\left\langle (c_j^{(k)})^{\#} (c_{j'}^{(k)})^{\#} \right\rangle_{\rho_{\psi_k}} & \mbox{if } j,j' \in \Lambda_k .\end{array} \right.
\end{equation}
The claim when $j$ and $j'$ are in the same block is clear, as in this case one checks from the definition of the local Jordan-Wigner systems that
\begin{equation}
 (c_j)^{\#} (c_{j'})^{\#} = \idty \otimes (c_j^{(k)})^{\#} (c_{j'}^{(k)})^{\#} \otimes \idty,
 \end{equation}
with the two identity operators acting on $\mathcal{H}_{\Lambda_1} \otimes \ldots \otimes \mathcal{H}_{\Lambda_{k-1}}$ and $\mathcal{H}_{\Lambda_{k+1}} \otimes \ldots \otimes \mathcal{H}_{\Lambda_m}$, respectively.

To see that distinct blocks $k \not= k'$ yield zero, we
have used that $\rho_{\psi_{k'}}$ is quasi-free and thus the expectation associated to $c_{j'}^*$ is zero.
\end{proof}

We now prove our main result.

\begin{proof}[Proof of Theorem~\ref{thm:DAL}]

An important first step is to write the eigenvectors $\psi_k$ of the restricted XY chains $H_{\Lambda_k}$ as eigenvectors of equivalent free Fermion systems:
As recalled in Section~\ref{sec:effham}, an ONB of eigenvectors $\psi_{\alpha}$, $\alpha \in \{0,1\}^n$, of $H$ is given by the Fermion basis (\ref{eq:PsiAlpha}). The same argument applies to each $H_{\Lambda_k}$, $k=1,\ldots,m$, giving Fermion bases
\begin{equation} \psi_{\alpha^{(k)}}, \quad \alpha^{(k)} \in \{0,1\}^{|\Lambda_k|}
\end{equation}
in $\mathcal{H}_{\Lambda_k}$, $k=1,\ldots,m$. We know that, almost surely, each $H_{\Lambda_k}$ has simple spectrum. This means that each of the eigenvectors $\psi_k$ of $H_{\Lambda_k}$ in the definition (\ref{eq:instate}) of $\rho$ must, up to a phase, be equal to one of the Fermion eigenvectors $\psi_{\alpha^{(k)}}$ and the corresponding eigenprojector equal to $\rho_{\alpha^{(k)}} = |\psi_{\alpha^{(k)}} \rangle \langle \psi_{\alpha^{(k)}}|$. This shows that
\begin{equation}
\rho = \rho_{\alpha} := |\psi_{\alpha} \rangle \langle \psi_{\alpha} |,
\end{equation}
where $\psi_{\alpha} := \psi_{\alpha^{(1)}} \otimes \ldots \otimes \psi_{\alpha^{(m)}}$, $\alpha := (\alpha^{(1)}, \ldots, \alpha^{(m)}) \in \{0,1\}^n$. Thus the claim (\ref{AreaLaw}) is equivalent to
\begin{equation} \label{AreaLaw2}
\mathbb{E} \left( \sup_{t,\alpha} \mathcal{E}((\rho_{\alpha})_t) \right) \le C < \infty
\end{equation}
uniformly in $n$, $a$ and $b$.

In the following we fix $t$ and $\alpha$ and abbreviate $\Gamma := \Gamma_{(\rho_{\alpha})_t}^{\mathcal{C}}$. Also, let $\Gamma_1$ be the restriction of $\Gamma$ to span$\{e_{2j-1}, e_{2j}: j \in \Lambda_0\}$. As explained at the beginning of the section, $\rho_t$ is quasi-free. Thus, by Lemma~\ref{lem:AR-S}, $\Gamma_1$ is the $2|\Lambda_0| \times 2|\Lambda_0|$ correlation matrix of the reduced state $(\rho_{\alpha})_t^1$ and
\begin{equation} \label{traceid}
\mathcal{E}((\rho_{\alpha})_t) = - \tr \Gamma_1 \log \Gamma_1.
\end{equation}
To find an upper bound for the latter, we follow an argument in the proof of Theorem~1.1 in \cite{AR-S}, which itself generalized an earlier argument in \cite{PasturSlavin}. We refer to these works for additional detail for some of the following.

In particular, the Fermionic correlation matrix $\Gamma_1$ has pairs of eigenvalues $\xi_j, 1-\xi_j$, $j=1,\ldots,|\Lambda_0|$, where $0\le \xi_j \le 1/2$. Thus, using (\ref{traceid}), the elementary bound
\begin{equation}
-x \log x + (1-x) \log (1-x) \le 2 \log 2 \cdot \sqrt{x(1-x)}, \quad 0<x<1,
\end{equation}
the Peierls-Bogoliubov inequality, and the elementary inequality  $\sqrt{x}+\sqrt{y}\leq \sqrt{2}\sqrt{x+y}$, one gets
\begin{equation} \label{onemore}
\mathcal{E}((\rho_{\alpha})_t) \le \sqrt{2} \log 2 \sum_{j\in \Lambda_0} \left( \tr (\Gamma_1 (\idty - \Gamma_1) )_{jj} \right)^{1/2}.
\end{equation}
Using that $\Gamma$ is an orthogonal projection, one can show, for $j\in \Lambda_0$,
\begin{equation}
(\Gamma_1 (\idty - \Gamma_1))_{jj} = \sum_{k\in \Lambda \setminus \Lambda_0} \Gamma_{jk} (\Gamma_{jk})^*.
\end{equation}
Inserting this into (\ref{onemore}) leads to
\begin{equation} \label{PSbound}
\mathcal{E} ( (\rho_{\alpha})_t) \leq  2\log 2 \sum_{\ell \in \Lambda_0}\sum_{\ell' \in \Lambda \setminus \Lambda_0} \| \Gamma_{\ell ,\ell'} \|
\end{equation}
for all $\alpha$ and $t$.

Thus we are led to having to bound the expectations of $\|\Gamma_{\ell,\ell'}\|$. By Lemma~\ref{lem:corr:tautC},
\begin{equation}
\Gamma = e^{-2itM} \Gamma_{\rho_{\alpha}}^{\mathcal{C}} e^{2itM}.
\end{equation}
This yields
\begin{equation}
\Gamma = e^{-2itM} \bigoplus_{k=1}^m \Gamma_{\rho_{\alpha^{(k)}}}^{\mathcal{C}^{(k)}} e^{2itM},
\end{equation}
where we have also used Lemma~\ref{lem:corinirho}. As recalled in Section~\ref{app:carqf}, the correlation matrices of the Fermion basis vectors are related to spectral projections of the effective Hamiltonian by (\ref{specproj}). When applied to the restrictions $H_{\Lambda_k}$ of the XY chain, this gives
\begin{equation} \label{specprojk}
 \Gamma_{\rho_{\alpha^{(k)}}}^{\mathcal{C}^{(k)}} = \chi_{\Delta_{\alpha^{(k)}}}(M_k), \quad k=1,\ldots,m.
 \end{equation}
Here $M_k$ is the effective Hamiltonian for $H_{\Lambda_k}$, i.e.\ the restriction of the block-Jacobi matrix (\ref{effHam}) to $\Lambda_k$, $\Delta_{\alpha^{(k)}}$ is a suitable subset of the spectrum of $M_k$ (defined in terms of the eigenvalues of $M_k$ in analogy to (\ref{Deltas})), and the right hand side of (\ref{specprojk}) denotes the spectral projection for $M_k$ onto $\Delta_{\alpha^{(k)}}$.

In summary, we find that
\begin{equation}\label{eq:Gamma}
\Gamma = e^{-2itM} \chi_{\alpha} e^{2itM},\ \text{where}\ \chi_{\alpha}:=\bigoplus_{k=1}^m \chi_{\Delta_{\alpha^{(k)}}}(M_k)
\end{equation}
Using (\ref{eq:Gamma}), it is clear that
\begin{equation}\label{eq:Gammahat:elements}
  \Gamma_{\ell, \ell'}=\sum_{z, z' \in \Lambda}(e^{-2itM})_{\ell, z}(\chi_{\alpha})_{z, z'}(e^{2itM})_{z', \ell'}.
\end{equation}
Let us denote by `$\sup$' the supremum over all $\alpha \in \{0,1\}^n$ and $t\in \R$.
The H\"older inequality implies
\begin{multline} \label{holder}
\mathbb{E} \left( \sup \|  \Gamma_{\ell, \ell'} \| \right) \leq  \sum_{z,z' \in\Lambda} \mathbb{E}\left(\sup \|(e^{-2itM})_{\ell, z}\|^3\right)^\frac{1}{3}
  \mathbb{E}\left(\sup \|(\chi_{\alpha})_{z, z'}\|^3\right)^\frac{1}{3} \times \\
\times \mathbb{E}\left(\sup \|(e^{2itM})_{z', \ell' }\|^3\right)^\frac{1}{3}
\end{multline}
Since all the operators on the right above have norm bounded by 1, we may estimate $\| \cdot \|^3 \leq \| \cdot \|$ and
then apply our decay assumption (\ref{ecorloc}) on eigencorrelator localization. Note that by assuming the validity of (\ref{ecorloc}) for {\it all} $n$ and by the translation invariance of the distribution of random parameters, (\ref{ecorloc}) also applies to the effective Hamiltonians $M_k$ of the subsystems. Furthermore, both, eigenprojections $\chi_{\Delta_{\alpha^{(k)}}}(M_k)$ as well as the time evolution $e^{\pm 2it M}$, fall within the set of functions $|g|\le 1$ covered by (\ref{ecorloc}). This leads to bounds for all the matrix elements appearing in (\ref{holder}). Inserting the result into (\ref{PSbound}) gives
\begin{eqnarray}
\lefteqn{\mathbb{E} \left( \sup \mathcal{E}( (\rho_{\alpha})_t) \right)} \\ & \leq & 2 \log 2 \cdot C^3 \sum_{\ell, \ell'} \sum_{z,z' }
\frac{1}{(1+ |\ell - z|)^{\beta /3}} \frac{1}{(1+ |z - z'|)^{\beta /3}} \frac{1}{(1+ |z' - \ell'|)^{\beta /3}}  \nonumber \\
& \leq & 2 \log 2 \cdot C^3 \cdot D^2 \sum_{\ell \in \Lambda_0} \sum_{\ell' \in \Lambda \setminus \Lambda_0} \frac{1}{(1+ | \ell - \ell'|)^{\beta/3}} \nonumber \\
& \leq & 4 \log 2 \cdot C^3 \cdot D^2 \sum_{j=1}^{\infty} \frac{j}{(1+j)^{\beta/3}} =: C' < \infty \nonumber
\end{eqnarray}
if $\beta >6$. Note that for the second bound above we used the following simple fact. For any $\beta >1$, there
is a positive number $D < \infty$ for which
\begin{equation}
\sum_{z \in \mathbb{Z}} \frac{1}{(1+ |x-z|)^{\beta}} \frac{1}{(1+|z-y|)^{\beta}} \leq D \frac{1}{(1+|x-y|)^{\beta}} \quad \mbox{for all } x,y \in \mathbb{Z} \, .
\end{equation}
For example, one may take $D =2^{\beta} \sum_{z} (1+|z|)^{-\beta}$. This completes the proof of (\ref{AreaLaw2}).
\end{proof}

We finally comment on how the above proof changes, in fact simplifies, under the special cases considered in Corollary~\ref{cor:updown} and Corollary~\ref{cor:AR-S}. In the case of initial condition $\rho = |e_{\alpha} \rangle \langle e_{\alpha}|$ given by up-down spins, $\Gamma$ in (\ref{eq:Gamma}) will be given as
\begin{equation}
\Gamma = e^{-2itM} N_{\alpha} e^{2itM},
\end{equation}
where $N_{\alpha}$ is a form of the number operator,
\begin{equation}
N_{\alpha} = \diag (N_{\alpha,j}: j=1,\ldots,n), \quad N_{\alpha,j} = \begin{pmatrix} \delta_{\alpha_j,0} & 0 \\ 0 & \delta_{\alpha_j,1}. \end{pmatrix}
\end{equation}
Then (\ref{eq:Gammahat:elements}) will read
\begin{equation}
\Gamma_{jk} = \sum_r (e^{-itM})_{jr} N_{\alpha,r} (e^{itM})_{rk}.
\end{equation}
As $\|N_{\alpha,r}\|=1$ for all $r$, this leads to an analogue of (\ref{holder}),
\begin{eqnarray}
\E \left( \sup_{t,\alpha} \| \Gamma_{jk} \|  \right) & \le & \sum_r \E \left( \sup_t \|(e^{-itM})_{jr}\| \|(e^{itM})_{r k}\| \right) \\
& \le & \sum_r \left( \E(\sup_t \|(e^{-itM})_{jr}\|^2 ) \right)^{1/2} \left( \E(\sup_t \|(e^{itM})_{r k}\|^2 ) \right)^{1/2}. \nonumber
\end{eqnarray}
This will give an area law requiring $\beta>4$ only.

Finally, for Corollary~\ref{cor:AR-S} we simply get
\begin{equation}
\Gamma = e^{-2itM} \chi_{\Delta_{\alpha}}(M) e^{2itM} = \chi_{\Delta_{\alpha}}(M).
\end{equation}
From here the proof proceeds without the need for using H\"older's inequality and gives an area law for $\beta>2$.

%
%
%
%
%

\section{On some basics concerning the CAR algebra and Quasi-Free States} \label{app:carqf}

\subsection{Basic concepts} A number of the results in this article depend on some general properties of quasi-free states defined
on a CAR algebra. For the convenience of the reader, we collect some of these facts. A more thorough
discussion may be found e.g. in Section 5.2 of \cite{BRvol2}.

In general, given a separable complex Hilbert space $\mathcal{K}$, there is a unique, up to $*$-isomorphism,
$C^*$-algebra $\mathcal{A} = \mathcal{A}( \mathcal{K})$ generated by the identity, $\idty$, and elements $c(f)$ and
$c^*(g)$ for $f,g \in \mathcal{K}$ which satisfy:
\begin{enumerate}
\item the mappings $f \mapsto c(f)$ is anti-linear and $g \mapsto c^*(g)$ is linear;
\item $\{ c(f), c(g) \} = \{c^*(f), c^*(g) \} = 0$ for all $f,g \in \mathcal{K}$;
\item $\{ c(f), c^*(g) \} = \langle f, g \rangle \idty $
\end{enumerate}
This algebra $\mathcal{A}$ is called the CAR algebra generated by these $c$-operators over $\mathcal{K}$.
A proof of this fact is given as Theorem 5.2.5 in \cite{BRvol2}.

We will mainly be concerned with the CAR algebra defined over the finite dimensional Hilbert space
$\mathcal{K}_n = \ell^2( \{1, \cdots, n \})$ for some integer $n \geq 1$. Let us denote this CAR algebra by $\mathcal{A}_n$.
An explicit representation of $\mathcal{A}_n$ is given by
\begin{equation}
c(f) = \sum_{j=1}^n \overline{f(j)} c_j \quad \mbox{and} \quad c^*(g) = \sum_{j=1}^n g(j) c_j^* \quad \mbox{for all } f,g \in \mathcal{K}_n
\end{equation}
where, for $1 \leq j \leq n$, the operators $c_j$ and $c_j^*$ are the Jordan-Wigner variables introduced
in (\ref{eq:cjs}). One readily checks that the above properties hold and it is clear that, with respect to the
canonical basis vectors $\{ \delta_j \}$ in $\mathcal{K}_n$, $c( \delta_j) = c_j$
and $c^*( \delta_j) = c_j^*$ for all $1 \leq j \leq n$.
For many calculations, it is convenient to define general linear combinations of these operators.
To this end, one sets
\begin{equation}
C(f,g) : = c(f) + c^*(g)  \quad \mbox{for all } f,g \in \mathcal{K}_n \, .
\end{equation}

A state $\omega$ on a CAR algebra $\mathcal{A}$ is a mapping $\omega : \mathcal{A} \to \mathbb{C}$
that is linear, positive, and normalized so that $\omega( \idty) = 1$.  The state $\omega$ on
a CAR algebra $\mathcal{A}$ is said to be {\it quasi-free} if for all $m \geq 1$ and all collections
$\{ f_j \}_{j=1}^m, \{ g_j \}_{j=1}^m \subset \mathcal{K}$ and $C_j := C(f_j,g_j)$ we have that
\begin{equation} \label{quasifree}
\omega \left( \prod_{j=1}^m C_j\right) = \left\{\begin{array}{ll}
                                                    0, & \hbox{if $m$ is odd,} \\
                                                    \displaystyle\sum_{k=2}^m (-1)^{k} \omega(C_k C_1) \omega \left(\prod_{\tiny{\begin{array}{c}
                                                  j\in\{2,\ldots,m\}\setminus\{k\}
                                                \end{array}}
}^m C_j\right), & \hbox{if $m$ is even.}
                                                  \end{array}
                                                \right.
\end{equation}
For even $m$ this is a version of Wick's rule and can be iterated to yield that general expectations $\omega (\prod_{k=1}^m C_k)$ are sums of products of pair expectations. The exact resulting expression is known as the pfaffian pf$[ \mathcal{C}]$ of the skew-adjoint $m \times m$-matrix $\mathcal{C} = \mathcal{C}(m)$ with entries
\begin{equation}
\mathcal{C}_{j,k} = \omega \left(C_j C_k \right) \quad \mbox{for all } 1 \leq j<k \leq m,
\end{equation}
and extended to be skew-adjoint. As the pfaffian of an odd-dimensional skew adjoint matrix is generally set to zero, one may write (\ref{quasifree}) as $\omega \left( \prod_{j=1}^m C_j  \right) = {\rm pf} [ \mathcal{C} ]$. Beyond convenience of notation, this will allow to exploit known facts on pfaffians below.

We will also be interested in dynamically evolved states. Let $\tau_t$ be a dynamics on $\mathcal{A}$,
i.e. a one-parameter group of $*$-automorphisms of $\mathcal{A}$. Often such a dynamics has the form
\begin{equation} \label{dyn}
\tau_t(A) = e^{itH}A e^{-itH} \quad \mbox{for all } A \in \mathcal{A} \mbox{ and all }  t \in \mathbb{R} \,
\end{equation}
where $H$ is some self-adjoint element of $\mathcal{A}$.
Generally when ${\rm dim}( \mathcal{K}) = n < \infty$, as will be the case in all our considerations,
we may identify $\mathcal{A} = \mathcal{A}_n := \mathcal{B}\left(\bigotimes_{j=1}^n \C^2\right)$, as is seen by choosing the representation of the CAR algebra by the Jordan-Wigner operators (\ref{eq:cjs}). Thus in this finite dimensional case, any state $\omega$ on this CAR algebra is in one to
one correspondence with a density matrix $\rho \in \mathcal{A}_n$, i.e.
\begin{equation}
\omega(A) = {\rm Tr}[A \rho] \quad \mbox{for all } A \in \mathcal{A}_n \, ,
\end{equation}
where $0 \leq \rho \leq \idty$ satisfies ${\rm Tr}[ \rho ] = 1$. For any $t \in \mathbb{R}$, such a state $\omega$
can be time-evolved (with respect to the dynamics $\tau_t$) by writing
\begin{equation} \label{eq:evolvedstate}
\omega_t(A) : = \omega( \tau_t(A)) \quad \mbox{for all } A \in \mathcal{A}_n \, .
\end{equation}
In this case,
\begin{equation}
\omega_t(A) = \omega( \tau_t(A)) = {\rm Tr}[ A \rho_t] \quad \mbox{with } \rho_t = e^{-itH} \rho e^{itH} \, ,
\end{equation}
and it is clear that the time-evolution of a state, i.e. $\omega_t$, is still a state.

\subsection{The main example} We now consider the examples of most interest in our work.
Let $H$ be the anisotropic XY-spin chain on the interval $[1,n]$, see (\ref{eq:anisoxychain}). As discussed in
Section~\ref{sec:effham}, a collection of orthonormal eigenvectors of $H$ can be
parametrized by $\alpha \in \{0,1\}^n$; in particular, they have the form:
\begin{equation}\label{psia}
\psi_{\alpha} = (b_1^*)^{\alpha_1} \ldots (b_n^*)^{\alpha_n} \Omega
\end{equation}
where the $b$-operators are as defined in (\ref{cov}), and $\Omega$ is the corresponding vacuum vector.
It is well-known that each of the states $\omega_{\alpha}$ defined by
\begin{equation}
\omega_{\alpha}(A) := {\rm Tr}[ A \rho_{\alpha}] \quad \mbox{for all } A \in \mathcal{A}_n
\end{equation}
with the density matrix $\rho_{\alpha} := | \psi_{\alpha} \rangle \langle \psi_{\alpha} |$ is quasi-free, see e.g. \cite{AR-S}. Thus
eigenstates of the XY model form a class of quasi-free states on $\mathcal{A}_n$. As such, in order
to calculate expectations in these eigenstates, one needs (in principle) only know the corresponding pair expectations.
These pair expectations are often organized in a correlation matrix:
\begin{equation} \label{evcorm}
\Gamma_{\rho_{\alpha}}^{\mathcal{C}} : = \omega_{\alpha} \left( \mathcal{C} \mathcal{C}^* \right)
\end{equation}
where $\mathcal{C} = (c_1, c_1^*, c_2,c_2^*, \cdots, c_n, c_n^*)^T$ is the column vector of
Jordan-Wigner variables discussed above,  $\mathcal{C}^*$ is the corresponding row vector, and
the right side of (\ref{evcorm}) is to be understood entry-wise:
$(\Gamma_{\rho_{\alpha}}^{\mathcal{C}})_{j,k} = \omega_{\alpha}( (\mathcal{C} \mathcal{C}^*)_{j,k})$.
The formula
\begin{equation} \label{specproj}
\Gamma_{\rho_{\alpha}}^{\mathcal{C}} = \chi_{\Delta_\alpha}(M)
\end{equation}
is also well-known to hold in the case that $M$ has simple spectrum, see e.g. \cite{AR-S}. Here the right side is defined through the functional calculus associated to
the self-adjoint operator $M$. In fact, $\chi_{\Delta_\alpha}(M)$ is the spectral projection onto the set
\begin{equation} \label{Deltas}
\Delta_{\alpha}  = \{ \lambda_j \, : \, \alpha_j = 0 \} \cup \{ - \lambda_j \, : \, \alpha_j =1 \} \, .
\end{equation}

As we have discussed above, see (\ref{BCdynamics}), the Heisenberg time evolution $\tau_t(A) = e^{itH} A e^{-itH}$ under the XY chain $H$ is governed by an effective
single-particle Hamiltonian, i.e.\ there is a self-adjoint $2n\times 2n$-matrix $M$ for which $\tau_t(\mathcal{C}) = e^{-2itM} \mathcal{C}$.  Identifying $f, g \in \mathcal{K}_n$ with row vectors $(f(1),\ldots,f(n))$ and $(g(1),\ldots,g(n))$ this is easily seen to give the following.
\begin{lem} \label{lem:qftime}
The Heisenberg dynamics of the operators $C(f,g)$ satisfy
\begin{equation}
\tau_t(C(f,g)) = C(f_t,g_t),
\end{equation}
where $(\bar{f}_t, g_t) = (\bar{f}, g) e^{-2itM}$.
\end{lem}

In this case, if $\omega$ is a quasi-free state, then for any $t \in \mathbb{R}$,  $\omega_t$ given by (\ref{eq:evolvedstate}) is quasi-free as well.
This follows from a direct calculation:
\begin{equation} \label{pfaff1}
\omega_t \left( \prod_{j=1}^m C(f_j, g_j) \right)  =  \omega \left( \prod_{j=1}^m \tau_t (C(f_j,g_j)) \right) =
\omega \left( \prod_{j=1}^m C( (f_j)_t, (g_j)_t) \right)  = {\rm pf}[ \mathcal{C}_t]
\end{equation}
where we have used the automorphism property of $\tau_t$ and the fact that $\omega$ is quasi-free.
Here for any $1 \leq j<k\leq m$,
\begin{equation} \label{pfaff2}
(\mathcal{C}_t)_{j,k} = \omega \left(C((f_j)_t, (g_j))_t C((f_k)_t, (g_k))_t \right) = \omega_t \left(C(f_j, g_j) C(f_k, g_k) \right) \, .
\end{equation}

\subsection{Products of quasi-free states are quasi-free} In this subsection, we prove the following lemma.

\begin{lem} \label{lem:addqf} For $i=1,2$, let $\rho_i$ be density matrices corresponding to quasi-free states on $\mathcal{A}_{n_i}$.
Then $\rho = \rho_1 \otimes \rho_2$ is the density matrix corresponding to a
quasi-free state on $\mathcal{A}_{n_1+n_2}$.
\end{lem}

In the proof of Theorem~\ref{thm:DAL} in Section~\ref{sec:DAL} above, we use this lemma for the case in which $\rho$ is a product of eigenstates of restrictions of the XY chain to subintervals of $[1,n]$. In this case, a more direct proof of Lemma~\ref{lem:addqf} follows from the fact that the product state can be considered as an eigenstate of an XY chain on $[1,n]$ (replace the interaction parameters $\mu_j$ by zero for all $j$ at the boundaries of subsystems). We nevertheless provide a proof of Lemma~\ref{lem:addqf} in its general form, as we consider this to be of some independent interest.

\begin{proof}
Let us denote by $\{ c_j \}_{j=1}^{n_1+n_2}$ the Jordan-Wigner variables associated to the chain of length $n_1+n_2$.
Our goal is to prove that: for any $m \geq 1$ and any collections $\{ f_j \}_{j=1}^m, \{ g_j \}_{j=1}^m \subset
\mathcal{K}_{n_1+n_2}:= \ell^2( \{ 1, \cdots, n_1+n_2\})$. The relation
\begin{equation} \label{prodqf}
\langle \prod_{j=1}^m C(f_j,g_j) \rangle_{\rho} = {\rm pf}( \mathcal{C} )
\end{equation}
holds, where the anti-symmetric $(m \times m)$ matrix $\mathcal{C}= \mathcal{C}(m)$ has entries
\begin{equation}
\mathcal{C}_{jk} = \langle C(f_j,g_j) C(f_k,g_k)\rangle_{\rho} \quad \mbox{for all } 1 \leq j< k \leq m
\end{equation}
and then extended by anti-symmetry.

The proof of this follows from a number of calculations. These calculations are driven by the
observation that the combined system (of size $n_1+n_2$) can be regarded as a combination
of the two smaller systems (of sizes $n_1$ and $n_2$ respectively) in the following concrete way.
With a slight abuse of notation, for all $1 \leq j \leq n_1+n_2$ relabel the Jordan-Wigner variables introduced above as
$c_j^{(12)}:=c_j$, indicating the combined system. Note that
\begin{equation} \label{sysc}
c_j^{(12)} = \left\{ \begin{array}{ll} c_j^{(1)} \otimes \idty, & 1 \leq j \leq n_1, \\
(\sigma^z)^{\otimes^{n_1}} \otimes c_{j}^{(2)} & n_1+ 1 \leq j \leq n_1+n_2 \end{array} \right.
\end{equation}
where now, for $i=1,2$, the operators $c^{(i)}_k$ denote the corresponding Jordan-Wigner variables of the sub-systems
of size $n_i$.

Identifying each $f \in \mathcal{K}_{n_1+n_2}$ with a sequence of coefficients, we write $f = f^{(1)} +f^{(2)}$ with $f^{(1)}$ the first
$n_1$ coefficients and $f^{(2)}$ the remaining $n_2$ coefficients. With this in mind, the first step in the proof of
(\ref{prodqf}) is to verify that
\begin{equation} \label{step1}
C^{(12)}(f,g) = C^{(1)}(f^{(1)}, g^{(1)}) \otimes \idty + ( \sigma^z)^{\otimes^{n_1}}\otimes C^{(2)}(f^{(2)}, g^{(2)})
\end{equation}
for all $f,g \in \mathcal{K}_{n_1+n_2}$. Next, we claim that for any $m \geq 1$, the matrix $\mathcal{C}^{(12)} := \mathcal{C}$
appearing on the right-side of (\ref{prodqf}) satisfies
\begin{equation} \label{step2}
\mathcal{C}^{(12)} = \mathcal{C}^{(1)} + \mathcal{C}^{(2)}
\end{equation}
in the sense that for any collections $\{ f_j\}_{j=1}^m , \{g_j \}_{j=1}^m \subset \mathcal{K}_{n_1+n_2}$
\begin{eqnarray} \label{mentries}
\mathcal{C}^{(12)}_{jk} & := & \langle C^{(12)}(f_j, g_j) C^{(12)}(f_k,g_k) \rangle_{\rho} \nonumber \\
& = &  \langle C^{(1)}(f_j^{(1)}, g_j^{(1)}) C^{(1)}(f_k^{(1)},g_k^{(1)}) \rangle_{\rho_1} + \langle C^{(2)}(f_j^{(2)}, g_j^{(2)}) C^{(2)}(f_k^{(2)},g_k^{(2)}) \rangle_{\rho_2}  \nonumber \\
& =: &  \mathcal{C}^{(1)}_{jk} + \mathcal{C}^{(2)}_{jk} \quad \mbox{for all } 1 \leq j<k \leq m.
\end{eqnarray}
Lastly we show directly that
\begin{equation} \label{step3}
\langle \prod_{j=1}^m C(f_j,g_j) \rangle_{\rho} = {\rm pf}( \mathcal{C}^{(1)} + \mathcal{C}^{(2)})
\end{equation}
which, using (\ref{step2}), verifies the claim (\ref{prodqf}).

Equation (\ref{step1}) follows directly from (\ref{sysc}) and the analogous formula which holds for $(c^{(12)}_j)^*:= c_j^*$.
For step 2, calculate the pair expectation of interest using (\ref{step1}). As is clear from (\ref{mentries}), the claim (\ref{step2}) amounts to the observation
that the cross terms vanish. Note that e.g.
\begin{multline}
\langle C^{(1)}(f_j^{(1)}, g_j^{(1)}) \otimes \idty \cdot (\sigma^z)^{\otimes^{n_1}} \otimes C^{(2)}(f_k^{(2)}, g_k^{(2)}) \rangle_{\rho} \\ =  \langle  C^{(1)}(f_j^{(1)}, g_j^{(1)}) (\sigma^z)^{\otimes^{n_1}} \rangle_{\rho_1} \cdot \langle  C^{(2)}(f_k^{(2)}, g_k^{(2)}) \rangle_{\rho_2}
\end{multline}
Since each $\rho_i$ is quasi-free, the factors on the right-side are both pfaffians. Each is zero
since the pfaffian of an odd-dimensional, anti-symmetric matrix is zero. (\ref{step2}) is proven.

We need only prove (\ref{step3}). Given a product as on the left of (\ref{step3}) denote by
\begin{equation}
A_j = C^{(1)}(f_j^{(1)}, g_j^{(1)}) \otimes \idty \quad \mbox{and} \quad B_j = ( \sigma^z)^{\otimes^{n_1}}\otimes C^{(2)}(f_j^{(2)}, g_j^{(2)}) \, .
\end{equation}
One readily checks that $\{A_j, B_{j'} \} =0$ for all $j, j'$.
Moreover, it is then clear from (\ref{step1}) that naive expansion yields
\begin{equation}
\prod_{j=1}^m C(f_j,g_j) = \sum_{\alpha \in \{0,1\}^m} \prod_{j=1}^m  A_j^{\alpha_j} B_j^{1- \alpha_j} \, .
\end{equation}
To each $\alpha$, we may associate a set $J( \alpha) \subset \{1, \cdots, m\}$ by taking $J( \alpha) := \{ j : \alpha_j =1\}$.
Using the support of the operators $A_j$ and $B_j$, the fact that the density matrix $\rho$ has a product structure,
and the fact that both $\rho_1$ and $\rho_2$ are quasi-free, it is clear that if $m$ is odd or $\alpha$ is such that
the cardinality of $J(\alpha)$, denoted by $|J( \alpha)|$, is odd, then
\begin{equation}
\langle \prod_{j=1}^m  A_j^{\alpha_j} B_j^{1- \alpha_j}\rangle_{\rho} = 0 \, .
\end{equation}
Now, a careful re-organizing of the remaining products, using the anti-commutation relations, shows that
\begin{equation}
\langle \prod_{j=1}^m  A_j^{\alpha_j} B_j^{1- \alpha_j}\rangle_{\rho} = (-1)^{\frac{|J(\alpha)|}{2} + \sum_{j \in J(\alpha)} j} \cdot \langle \prod_{j\in J(\alpha)}C^{(1)}(f_j^{(1)}, g_j^{(1)}) \rangle_{\rho_1}  \cdot
\langle \prod_{j\in J(\alpha)^c}C^{(2)}(f_j^{(2)}, g_j^{(2)}) \rangle_{\rho_2}.
\end{equation}
As a result (\ref{step3}) follows from the pfaffian formula
\begin{equation}
{\rm pf}[A+B] = \sum_{\stackrel{J \subset \{ 1, 2, \cdots, m\}:}{|J| {\tiny \mbox{ is even}}}} (-1)^{ \frac{|J|}{2} + \sum_{j \in J} j  } \cdot {\rm pf}[A_J] {\rm pf}[B_{J^c}]
\end{equation}
a proof of which one can find in \cite{Muir}.
\end{proof}


\section*{Acknowledgements} G.\ S.\ would like to thank the Isaac Newton Institute for Mathematical Sciences, Cambridge, for support and hospitality during the program {\it Periodic and Ergodic Spectral Problems} in Spring 2015 where part of the work on this paper was done. B.N. acknowledges the stimulating environment and warm hospitality at the Erwin Schr\"odinger International Institute for Mathematical Physics, Vienna during the program {\it Quantum Many-Body Systems, Random Matrices, and Disorder}, July 2015. R. S. would like to thank
the UC Davis Mathematics Department for their hospitality during his recent sabbatical when work on this project began. This research was supported in part by the National Science Foundation under Grants DMS-1069320 (G.S.) and DMS-1515850 (B.N.), and by a grant from the Simons Foundation (\#301127 to R.S.). Finally, we would like to thank a referee for pointing out the interesting extensions (\ref{eq:SWbound1}) and (\ref{eq:SWbound2}) of Theorem~\ref{thm:PNF}.

\end{document}